\documentclass[a4paper,11pt]{article}
\usepackage{amsmath, amsthm, amssymb}
\usepackage[english]{babel}
\usepackage[latin1]{inputenc}
\usepackage{pdfsync}
\usepackage[T1]{fontenc}
\newtheorem{prop}{Proposition}

\title{\huge\textbf{
Finite reduction and Morse index 
estimates for mechanical systems}
}
\author{\hspace{-1cm}{\sc Franco Cardin${\, }^{1}$\quad Giuseppe De Marco${\,
}^{1}$ \quad Alessandro Sfondrini${\,
}^{2,3}$                  }
\vspace{0.3cm}
\\
\hspace{-1cm}${\ }^{1}$ Dipartimento di Matematica Pura ed Applicata,
Universit\`a di Padova,\\
\hspace{-1cm}Via Trieste, 63 - 35121 Padova, Italy
\\ 
\hspace{-1cm}${\ }^{2}$ 
Scuola Galileiana di Studi Superiori,
Universit\`a di Padova,\\
 \hspace{-1cm}Via G. Prati, 19 - 35122 Padova, Italy\\
 \hspace{-1cm}${\ }^{3}$ Dipartimento di Fisica ``G. Galilei'', Universit\`a di Padova, \\
 \hspace{-1cm}Via F. Marzolo, 8 - 35131 Padova, Italy
}
\begin{document}
\maketitle

\begin{abstract}\par\noindent
A simple version  of exact finite dimensional reduction for the variational setting of mechanical systems  is presented. It is  worked out by means of a thorough global version of the implicit function theorem for monotone operators. Moreover, the Hessian of the reduced function preserves all the relevant information of the original one, by  Schur's complement, which spontaneously appears in this context.  Finally, the results are straightforwardly extended to the case of a Dirichlet problem on a bounded domain.\\
\textsc{Keywords:}  Calculus of Variations, Finite-Dimensional Reduction, Morse index.
\\
{AMS subject classifications: 70H25,
37B30,
35J61.
 }

\end{abstract}

\section{Introduction}
Conservative classical mechanics can be regarded as a collection of variational problems: a mechanical path will be a curve which makes stationary the action functional. Such a critical path will be a local minimum for the action in the uniform topology in the space of the admissible (w.r.t. the boundary conditions) paths if there are no conjugate points along it; this condition is known to be related to the positive-definiteness of the second variation --the Hessian-- of the action functional. A precursor was
Jacobi who  proposed an early theory of the second variation in 19th century. But the very notion of conjugate point first appeared in Bliss \cite{bliss}  --see e.g. the historical notes by Carath\'eodory \cite{Cara} and Fraser \cite{fraser}. The definitive settlement lays on the well known framework of the Morse theory, \cite{morse, milnor, duistermaat}.

In practice, one is often forced to solve the variational problem numerically. It turns out that under certain hypotheses it is sufficient to solve a reduced problem in a subspace of the space of admissible paths, involving a \emph{finite} number of equations, and recover the very same solution of the full problem. Such a technique is sometimes referred to as an \emph{exact finite reduction} method. An early approch of this kind can be recognized in \cite{MarPro} where it was employed in the framework of infinite dimensional Morse theory. Later, two different reduction schemes have been proposed, both initially used to give new versions and proofs for the `Poincar\'e last geometric theorem': the broken geodesic method by Chaperon \cite{ChZ,Chap1, CH87, CH89, Chap2} and the Amann-Conley-Zehnder reduction \cite{AmZeh, {ConZeh83}, ConZeh}. We will focus on the latter approach, a concise exposition of which can be also found in \cite{Viterbo}. Such a method can be employed in many systems, such as fluid \cite{CarTeb} and molecular dynamics \cite{TurPasCar} and field theory \cite{cardin}. 

We propose a simple exact reduction theorem for mechanical Lagrangians with fixed endpoints as boundary conditions, giving an estimate of the optimal dimension of the reduced space (the {\it cutoff} $N$) and comparing our results for the cutoff with the ones obtained by means of fixed point techniques \cite{CarMar}. We then show that the Morse index of the solution can be easily expressed in terms of numerically computable quantities; here, as in many other situations \cite{Carl,Co}, the linear algebra trick known as `Schur complement' spontaneously appears. The results are straightforwardly extended --\emph{Sect. 5}--  to the case of Dirichlet problem on a bounded $\Omega\subset\mathbb{R}^m$.
\section{Setting}
We consider a mechanical Lagrangian,
\begin{equation}\mathcal{L}(q,\dot{q})=\frac{1}{2}|\dot{q}|^2-V(q)
\end{equation}
where the potential energy $V\in\mathcal{C}^2(\mathbb{R}^n,\mathbb{R})$ satisfies
\begin{equation}\sup_{q\in\mathbb{R}^n}\left|V''(q)\right|=C<\infty
\label{hess}
\end{equation}
so that the Hessian of the Lagrangian is bounded:
\begin{equation}\sup_{(q,\dot{q})\in\mathbb{R}^{2n}}\left|  \mathrm{Hess}[\mathcal{L}](q,\dot{q})\right|=\max\{1,C\}\equiv \tilde{C}
\end{equation}

One can now construct the action functional acting on the curves of the  Sobolev space $H^1$ with fixed endpoints: let
\[\Gamma\equiv\left\{\gamma\in H^1([0,T],\mathbb{R}^n),\ \gamma(0)=q_0,\,\gamma(T)=q_T\right\}\]
\[\tilde{L}: \Gamma\ni\gamma(\cdot)\mapsto\int_0^T\mathcal{L}\left(\gamma(t),\dot{\gamma}(t)\right)dt\in\mathbb{R}\]
Observe that the affine space $\Gamma$ can be parametrized 
\begin{equation}
\gamma(t)=q_0+\frac{q_T-q_0}{T}t+c(t),\ \ \ c\in \mathcal{H}\equiv H^1_0([0,T],\mathbb{R}^n)
\end{equation}
where $\mathcal{H}$ is the \lq loop space based at $0$\rq\ ($c(0)=c(T)=0$).
We define the action functional $L=\tilde{L}\circ\iota$ on $\mathcal{H}$:
\[L(c)\equiv \tilde{L}(\iota(c)),\ \ \ \iota(c)(t)=q_0+\frac{q_T-q_0}{T}t+c(t)\]

It is easy to see that $L$ is twice continuously differentiable in the Fr\'echet sense. First we have, for every $c,h \in \mathcal{H}$
\begin{equation}
dL(c)[h]=\left\langle \nabla L(c), h\right\rangle=\int_0^T \left(\dot{c}(t)\cdot\dot{h}(t)- V'(\iota(c(t)))[ h(t)]\right)dt
\end{equation}
where for every $c\in\mathcal{H}$ we can interpret $dL(c)\in\mathcal{H}^*$ as an element $\nabla L(c)\in\mathcal{H}$ by Riesz representation theorem.  Next, we find that the \lq second variation\rq\ is the bilinear form in $h,k\in\mathcal{H}$:
\begin{equation}
d^2L(c)\,[h,k]=\int_0^T \left(\dot{h}(t)\cdot\dot{k}(t)-V''(\iota(c(t)))\, [h(t),k(t)]\right)dt
\end{equation}
which can also be considered as a self-adjoint linear operator on the Hilbert space $\mathcal{H}$. We represent every $c\in \mathcal{H}$ by means of  the orthogonal basis $(\varphi^j_k(t))_{k\ge 1,\,j\in\{1,\dots,n\}}$, where for each $k>0$ and $j\in\{1,\dots,n\}$ the function $\varphi^j_k :[0,T]\to\mathbb{R}^n$ has all components $0$, except for the $j-$th which  is $\sqrt{2/T}\sin(k\,t/T)$; to simplify notation, we set $\varphi_k=(\varphi_k^1,\dots,\,\varphi_k^n)$,  $c^{(k)}=(c^1_k,\dots,c^n_k)$ and $c^{(k)}\cdot\varphi_k=\sum_{j=1}^nc^j_k\,\varphi^j_k$ so that we may write
\begin{equation}
c(t)=\sum_{k>0} c^{(k)}\cdot \varphi_k(t);\,\,
 \left\|c\right\|^2_{H^1_0}:=\int_0^T|\dot c(t)|^2\,dt=\sum_{k>0}\frac{\pi^2}{T^2}k^2\, |c^{(k)}|^2;
\end{equation}
 notice that then $\int_0^T|c(t)|^2\,dt=\sum_{k>0}|c^{(k)}|^2$.

\section{Exact finite reduction}
For any given natural number $N$, using the ordered basis $\varphi^j_k$, we can decompose $\mathcal{H}$ in two orthogonal subspaces, one of which is finite dimensional (of dimension $nN$):
\[\mathcal{H}=\mathcal{U}_N\oplus\mathcal{V}_N,\ \ \ \mathcal{U}_N=\mathrm{span}\{\varphi^j_k,\  1\le k\leq N,\,j\in\{1,\dots,n\}\} \]
Adopting a common notation for exact reduction techniques, we define the orthogonal  projectors $\mathbb{P}_N, \mathbb{Q}_N$ onto $\mathcal{U}_N$ and $\mathcal{V}_N$, respectively. $\mathcal{U}_N^*$ will be identified with the subspace of $\mathcal{H}^*$ consisting of functionals vanishing on $\mathcal{V}_N$, and similarly $\mathcal{V}_N^*$. When considering the differential of e.g. $L$ on $\mathcal{U}_N\oplus\mathcal{V}_N$, we will write $dL=\partial_{\mathcal{U}_N}L+\partial_{\mathcal{V}_N}L$ with $\partial_{\mathcal{U}_N}L\in\mathcal{U}_N^*$ and $\partial_{\mathcal{V}_N}L\in\mathcal{V}_N^*$. By Riesz theorem they will be identified with $\mathbb{P}_N\nabla L$ and $\mathbb{Q}_N\nabla L$ respectively. When interpreting $d^2L$ as a symmetric linear self-operator of $\cal H$, we will write \[d^2L=\partial_{\mathcal{U}_N}^2L+\partial_{\mathcal{V}_N}\partial_{\mathcal{U}_N}L+\partial_{\mathcal{U}_N}\partial_{\mathcal{V}_N}L+\partial_{\mathcal{V}_N}^2L\]
meaning $\partial_{\mathcal{U}_N}^2L=\mathbb{P}_N\, d^2L\,\mathbb{P}_N$, $\partial_{\mathcal{U}_N}\partial_{\mathcal{V}_N}L=\mathbb{Q}_N\,d^2L\mathbb{P}_N$, and so on. We will also denote by the same symbol operators such as $\partial_{\mathcal{U}_N}^2L$ on $\mathcal{H}$ and the one that it induces on $\mathcal{U}_N$.

In the following we use $\lfloor x\rfloor$ to denote the largest integer not greater than the real number $x$.

\begin{prop}For $N\ge \left\lfloor \frac{T\,\sqrt C}{\pi} \right\rfloor$  and any fixed $\bar{u}\in\mathcal{U}_N$, the {\rm restricted nonlinear operator}
\[ F(\bar u,\cdot)\equiv\mathbb{Q}_N\nabla L (\bar{u}+\cdot): \mathcal{V}_N\longmapsto\mathcal{V}_N\]
is strongly monotone.\label{monot}\end{prop}
\begin{proof}By strongly monotone we mean that there exists some ${\mu}_{\bar u} =\mu>0$ such that for any $v_1,v_2\in\mathcal{V}_N$ 
\begin{equation}\left\langle \mathbb{Q}_N\nabla L (\bar{u}+v_2)-\mathbb{Q}_N\nabla L(\bar{u}+v_1),v_2-v_1\right\rangle\geq \mu \left\|v_2-v_1\right\|^2_{{\mathcal V}_N}
\end{equation}
We have that 
\[\left\langle \mathbb{Q}_N\nabla L (\bar{u}+v_2)-\mathbb{Q}_N\nabla L(\bar{u}+v_1),v_2-v_1\right\rangle=\]
\[=\int_0^1 d^2L(\bar{u}+v_1+\lambda(v_2-v_1))\,[v_2-v_1,v_2-v_1]d\lambda\geq\]
\[\geq \int_0^1\int_0^T\left(|\dot{v}_2-\dot{v_1}|^2-\sup_{q\in\mathbb{R}^n}|V''(q)|\,|v_2-v_1|^2\right)dt\,d\lambda,\]
by setting $v\equiv v_2-v_1$,
\[=\int_0^T\left(|\dot{v}|^2-C\,|v|^2\right)dt=\sum_{k>N}\left(\frac{\pi^2 k^2}{T^2}-C\right)\,|v^{(k)}|^2.
\]
We wish to prove that there exist a positive integer $N$ and $\mu>0$ such that, for every $v\in \mathcal V_N$ we have
\[\sum_{k>N}\left(\frac{\pi^2 k^2}{T^2}-C\right)\,|v^{(k)}|^2\ge\mu\,\|v\|_{H^1_0}^2=\mu\,\sum_{k>N}\frac{\pi^2 k^2}{T^2}\,\,|v^{(k)}|^2;\]
this is plainly impossible if for some $k>N$ we have $\pi^2k^2/T^2-C\le0$. 
If $N=\lfloor T\,\sqrt{C}/\pi\rfloor$, we see that $(\pi\,k/T)^2-C>0$ for 
every integer $k>N$; let 
\[\mu=\min_{k> N}\left\{((\pi\,k/T)^2-C)/(\pi\,k/T)^2\right\}=1-C\,T^2/(\pi\,(N+1))^2>0;\]
 clearly, $N$ and $\mu$ are as required.
\end{proof}
 The {\em cutoff $N=\lfloor T\,\sqrt{C}/\pi\rfloor$}    measures, in a sense, the minimum amount of information  needed to describe the solution of the variational problem, and also provides an upper bound to the degeneracy of the number of solutions, as we shall see later. Hereafter $N$ will be at least as large as $\lfloor T\,\sqrt{C}/\pi\rfloor$.

From the above calculation it also easily  follows that $\partial^2_{\mathcal{V}_N}L$ is a uniformly positive definite self adjoint linear operator on $\mathcal{V}_N$ and thus a Hilbert spaces isomorphism.

We are now ready to prove the key statement for the exact finite reduction:
\begin{prop}There exists a continuously Fr\'echet differentiable map
 $\tilde{v}:\mathcal{U}_N\mapsto\mathcal{V}_N$  such that, for any $u\in {\mathcal U}_N$,
\begin{equation}
\mathbb{Q}_N \nabla L\left (u+\tilde v\left(u\right)\right)=0,\ \ {\rm or}\ \
\left\langle \nabla L\left(u+v\right),\delta v\right\rangle=0 \ \ \forall \delta v\in\mathcal{V}_N \ \Leftrightarrow  \ v=\tilde{v}(u)
\label{vtilda}
\end{equation}
\end{prop}
\begin{proof}Strong monotonicity of the mapping $F(u,\cdot):\mathcal V_N\to\mathcal V_N$ implies that this map is a self--homeomorphism of $\mathcal V_N$ (see e.g. \emph{Theorem 11.2} in \emph{Chap. 3} of \cite{deimling}; but stronger hypotheses allow us a much simpler proof than the one given there, see \emph{Appendix}). If we rewrite $(\ref{vtilda})$ as $F(u,v)=0$, we immediately have that for  each fixed $u\in\mathcal{U}_N$ there is a  unique solution $\tilde{v}(u)\in\mathcal{V}_N$. That $u\mapsto \tilde v(u)$ is continuously differentiable follows from the implicit function theorem: as remarked above  $\partial^2_{\mathcal{V}_N}L$ is uniformly positive definite and thus a Hilbert spaces isomorphism; and it is immediate to see that this linear map is also the partial differential $\partial_{\mathcal V_N}F(u,v)$ on the subspace $\mathcal V_N$.\end{proof}

\par\noindent
{\bf Remark\,1.}\ To obtain an explicit expression for $\tilde{v}$ one can proceed as follows: first one has to estimate $C<\infty$, and find the corresponding value of $N=\lfloor T\,\sqrt{C}/\pi\rfloor$ as obtained in \emph{Proposition \ref{monot}}. Then by Newton's method one can solve recursively for each fixed $u$
\[\tilde{v}_{n+1}(u)=\tilde{v}_n(u)-\left[\partial_{\mathcal V_N}F(u,\tilde{v}_n(u))\right]^{-1}\cdot F(u,\tilde{v}_n(u))\] 
which converges to the unique fixed point $\tilde{v}(u)$. Observe that the inverse $\partial_{\mathcal V_N}F^{-1}$ is a bounded operator on $\mathcal{V}_N$; furthermore, the convergence is faster when $\partial_{\mathcal V_N}F$ is far from zero, i.e. when one choses a large $N$ (this also happens in \cite{TurPasCar}). In practice to solve numerically the recursion equation it will be necessary to consider a truncation of $\tilde{v}_n$. One can avoid the inversion of $\partial_{\mathcal V_N}F$ by  suitable approximate methods, e.g. \cite{broyden,kvaalen}.
\\
\par\noindent
{\bf Remark\,2.}\  Once the function $\tilde v$ is known, solving fully the variational problem is equivalent to finding all the roots $u$ of  the finite dimensional system of equations
\begin{equation}
\mathbb{P}_N  \nabla L\left (u+\tilde v\left(u\right)\right)=0,\ \ {\rm or}\ \
\left\langle \nabla L\left(u+\tilde v(u)\right),\delta u\right\rangle=0 \ \ \forall \delta u\in\mathcal{U}_{N} \label{vtilda1}
\end{equation}
By a straightforward argument --see (\ref{dS}) below-- we see that
these equations are precisely the comeout from stationarizing the {\it reduced action functional} $S$: 
\begin{equation}\label{vtilda2} 
\frac{\partial S}{\partial u}(u)=0
\end{equation}
where
\begin{equation}\mathcal{U}_{N}\cong\mathbb{R}^N \ni u\mapsto S(u)\equiv L(u+\tilde{v}(u))\in\mathbb{R}
\end{equation}
These equations can be written explicitly in term of the Fourier components of $u=(u^j_k)_{1\le k\le N,\,j\in\{1,\dots,\,n\}}$\begin{equation}
\frac{\partial S}{\partial u^j_k}=0\ \ \Leftrightarrow\ \ \ \frac{\pi^2k^2}{T^2}u^j_k=\left[\nabla V\left(\iota(u+\tilde{v}(u))\right)\right]^j_k, \ \ \ 1\leq k\leq N,\,j\in\{1,\dots,\,n\}
\end{equation}

Finally, to any root $u$ of (\ref{vtilda2}), there corresponds the admissible mechanical curve $\gamma=\iota(u+\tilde v(u))$.
\\
\par\noindent

We conclude this section by comparing the minimal size of the cutoff $N$ with the known results obtained by other techniques; in particular, in \cite{CarMar} it has been shown by a fixed point argument that (in a more general case) it is sufficient to take $\tilde{N}$ such that
\begin{equation}\frac{\tilde{C}T}{2\pi\tilde{N}}\left(1+\sqrt{2\tilde{N}}\right)=\alpha<1
\label{boundcardin}
\end{equation}
It is easy to prove that  if $\tilde{N}$ verifies the preceding inequality then it is larger than $\lfloor T\,\sqrt{C}/\pi\rfloor$: in fact, observing   that since $\tilde N\ge1$ we have $(1+\sqrt{2\tilde N})/2>1$, and recalling that $\tilde{C}=\max\{1,C\}$
\[
\dfrac\pi T\,\tilde{N}>\tilde{C}\,\dfrac{1+\sqrt{2\tilde{N}}}2>\tilde{C}\ge\sqrt{\tilde{C}}\ge\sqrt{C},\]
hence
\[\tilde{N}>T\,\sqrt C/\pi\ge\lfloor T\,\sqrt{C}/\pi\rfloor.\]

Thus our estimate of the cutoff $N$  improves   the preceding one, at least in the mechanical case.

\section{Morse index}
For a Lagrangian which is positive definite in the velocities, the Morse index, being the maximal dimension of a negative subspace of $d^2 L(c)$, is finite and equal to the number of conjugate points counted with their multiplicity \cite{morse, milnor, duistermaat} along the critical path $\gamma=\iota(c), c=u+\tilde v(u), \frac{\partial S}{\partial u}(u)=0$. Recall that the conjugate points are the obstruction for such a path to be a local minimum, as mentioned in the introduction. 

We want to relate this number to the index of the \lq finite Hessian\rq\ $d^2S$.

\begin{prop}The Hessian of the reduced functional $S(u)=L(u+\tilde{v}(u))$ on $\mathcal{U}_N$ is
\begin{equation}
d^2S(u)=\partial_{\mathcal{U}_N}^2L(c)-\partial_{\mathcal{V}_N}\partial_{\mathcal{U}_N}L(c)\cdot\left[\partial_{\mathcal{V}_N}^2L(c)\right]^{-1}\cdot\partial_{\mathcal{U}_N}\partial_{\mathcal{V}_N}L(c)
\end{equation}\label{13}
where $c=u+\tilde v(u)$.
\end{prop}

\begin{proof}Differentiating $S$ once, we have
\begin{equation}\label{dS}
dS(u)=\partial_{\mathcal{U}_N}L(u+\tilde{v}(u))+\partial_{\mathcal{V}_N}L(u+\tilde{v}(u))\cdot\frac{\partial\tilde{v}}{\partial u}(u)=\partial_{\mathcal{U}_N}L(u+\tilde{v}(u))
\end{equation}
because $\partial_{\mathcal{V}_N}L(u+\tilde{v}(u))=0$ identically by definition of $\tilde{v}$. Differentiating this last identity we get
\[0=\partial_{\mathcal{U}_N}\partial_{\mathcal{V}_N}L(u+\tilde{v}(u))+\partial_{\mathcal{V}_N}^2L(u+\tilde{v}(u))\cdot\frac{\partial\tilde{v}}{\partial u}(u)\]
which, since $\partial_{\mathcal{V}_N}^2L(u+\tilde{v}(u))$ is strictly positive, gives
\begin{equation}\label{dvdu}
\frac{\partial\tilde{v}}{\partial u}(u)=-\left[\partial_{\mathcal{V}_N}^2L(u+\tilde{v}(u))\right]^{-1}\cdot\partial_{\mathcal{U}_N}\partial_{\mathcal{V}_N}L(u+\tilde{v}(u))
\end{equation}
Differentiating (\ref{dS}) once more, we get
\[d^2S(u)=\partial_{\mathcal{U}_N}^2L(u+\tilde{v}(u))+\partial_{\mathcal{V}_N}\partial_{\mathcal{U}_N}L(u+\tilde{v}(u))\cdot\frac{\partial\tilde{v}}{\partial u}(u)\]
that together with (\ref{dvdu}) gives the thesis.
\end{proof}

The Hessian,  on $\mathcal{H}=\mathcal{U}_N\oplus\mathcal{V}_N$, can be written in the block form
\[d^2 L=\left(\begin{array}{cc}
\partial_{\mathcal{U}_N}^2L & \partial_{\mathcal{V}_N}\partial_{\mathcal{U}_N}L\\
\\
\partial_{\mathcal{U}_N}\partial_{\mathcal{V}_N}L &\partial_{\mathcal{V}_N}^2L
\end{array}\right)\]
where $(\partial_{\mathcal{U}_N}\partial_{\mathcal{V}_N}L)^*=\partial_{\mathcal{V}_N}\partial_{\mathcal{U}_N}L$.

\begin{prop}The Morse index $i_L(u+\tilde{v}(u))$ of a solution equals the index of $S$ at $u$:
\begin{equation}i_L(u+\tilde{v}(u))= i_S(u)
\label{indexest}
\end{equation}
The same is true for the nullity.
\end{prop}
\begin{proof}
A standard technique (see e.g.\cite{Carl, Co}) to block--diagonalize the bilinear form $d^2L$ is by means of the block matrix
\[T=\left(\begin{array}{cc}
\mathbb{I}_{\mathcal{U}_N} & \mathbb{O}\\ 
\\
-\left[\partial_{\mathcal{V}_N}^2L\right]^{-1}\cdot\partial_{\mathcal{U}_N}\partial_{\mathcal{V}_N}L &\mathbb{I}_{\mathcal{V}_N}
\end{array}\right)\]
which  clearly defines a bounded linear operator on $\mathcal{H}=\mathcal U_N\oplus\mathcal V_N$, with bounded inverse
\[T^{-1}=\left(\begin{array}{cc}
\mathbb{I}_{\mathcal{U}_N} & \mathbb{O}\\
\\
\left[\partial_{\mathcal{V}_N}^2L\right]^{-1}\cdot\partial_{\mathcal{U}_N}\partial_{\mathcal{V}_N}L &\mathbb{I}_{\mathcal{V}_N}
\end{array}\right)\]The same holds for the adjoint $T^*$. We can thus use them to perform a change of basis on $\mathcal{H}$, so that $d^2L$ takes the block-diagonal form
\[T^*\,d^2L\,T=\left(\begin{array}{cc}
\partial_{\mathcal{U}_N}^2L-\partial_{\mathcal{V}_N}\partial_{\mathcal{U}_N}L\cdot\left[\partial_{\mathcal{V}_N}^2L\right]^{-1}\cdot\partial_{\mathcal{U}_N}\partial_{\mathcal{V}_N}L & \mathbb{O}\\
\\
\mathbb{O} &\partial_{\mathcal{V}_N}^2L
\end{array}\right)\]
It follows immediately that  negative and null spaces of $d^2L(u+\tilde{v}(u))$ and $d^2S(u)$ are isomorphic {\sl via} $T$.
\end{proof}
\par\noindent
 {\bf Remark\,3.}\ Contained in the previous argument is the a priori upper bound for the index:
 \begin{equation}
 i_L\leq \dim {\mathcal U}_N=nN
 \end{equation}

\section{Dirichlet problem}
Generalizing the preceding one dimensional case, 
we consider an open bounded $\Omega\subset\mathbb{R}^m$ with regular boundary $\partial\Omega$; in this case the Hilbert space we need is $\mathcal{H}\equiv H^1_0 (\Omega, \mathbb{R})$, the closure of $\mathcal{C}^\infty_c$ in $H^1(\Omega)$. The variational problem $\nabla L(\phi)=0,   \phi\in \mathcal{H}$, for the functional
\begin{equation}
L(\phi)=\int_\Omega\left(\frac{1}{2}\left|\nabla \phi(x)\right|^2-V(\phi(x))\right)dx
\end{equation}
can be regarded as a weak formulation of the Dirichlet problem
\[\Delta \phi=-V'(\phi),\ \ \ \ \phi|_{\partial\Omega}=0\]
In this case the basis vectors will be the eigenfunctions of the   Laplacian; as customary we arrange the distinct eigenvalues of the opposite of the Laplacian in  increasing order, with each $\lambda_k$ repeated according to multiplicity: $0<\lambda_1<\lambda_2\le \lambda_3\dots$,  and $\lambda_k\to+\infty$ as $k\to\infty$; 
an orthogonal basis will be $\{\varphi_k: \ \ k\ge 1\}$, where $-\Delta \varphi_k=\lambda_k\, \varphi_k$ and 
$\varphi_k|_{\partial\Omega}=0
$.
 Of course all this  depends on $\Omega$. For every integer $N\ge1$ we define
\[\mathcal{U}_N=\mathrm{span}\left\{\varphi_k\in\mathcal{H}\equiv H^1_0 (\Omega, \mathbb{R}):\ k\leq N\right\}\equiv\mathbb{P}_N\mathcal{H}\]
and $\mathcal{V}_N\equiv\mathbb{Q}_N\mathcal{H}$ is the orthogonal complement of $\mathcal U_N$.   Assuming  that 
\begin{equation}
\sup_{\phi\in\mathbb{R}}|V''(\phi)|=C<+\infty
\end{equation}
we can  again write (recalling that $\int_\Omega |\nabla v|^2=-
\int_\Omega v\Delta v $)
\[\left\langle \mathbb{Q}_N\nabla L (\bar{u}+v_2)-\mathbb{Q}_N\nabla L(\bar{u}+v_1),v_2-v_1\right\rangle\geq\]
\[\geq \int_\Omega\left(|\nabla (v_2-v_1)|^2-C\,|v_2-v_1|^2\right)dx\equiv\int_\Omega\left(|\nabla v|^2-C\,|v|^2\right)=\]
\[=\sum_{k>N}\left(\lambda_k-C\right)|v^{(k)}|^2\]
Pick the smallest integer $N$ such that $\lambda_{N+1}>C$ and set 
\[\mu=\min\left\{(\lambda_k-C)/\lambda_k:\,k>N\right\}\left(=1-C/\lambda_{N+1}\right)
;\]
we have that 
$\mu>0$ and as before
 \[\sum_{k>N}\left(\lambda_k-C\right)|v^{(k)}|^2\geq
 \mu\sum_{k>N}\lambda_k\,|v^{(k)}|^2=
  \mu\, || v||^2_{{\mathcal V}_N}
 ,\]
 giving the required strong monotonicity. In particular, one gets that the Morse index $i_L(c)$ of any solution $c=u+\tilde{v}(u)$ equals its reduced index $i_S(u)$, which is always bounded from above by the dimension of ${\mathcal U}_N$.
 \\
\par\noindent
 {\bf Remark\,4.}\ Recall that  for $m=1$ and vector valued functions in $\mathbb{R}^n$ we found that the minimal dimension of $\mathcal{U}_N$ to achieve the reduction was $nN$, with $N=\lfloor T\,\sqrt{C}/\pi\rfloor$. 
 This dimension is exactly the number of the eigenvalues of $-\Delta$ not larger than $C$; when $m\geq 1$ it  grows as $C^{m/2}$ with $C\to \infty$ as stated by Weyl's law \cite{levitan}:
\begin{equation}
\mathrm{dim}(\mathcal{U}_N)= \frac{\mathrm{vol}({B}_m)}{(2\pi)^m}\, \mathrm{vol}(\Omega)\,C^{\frac{m}{2}}+O(C^{\frac{m-1}{2}})
\end{equation}
where ${B}_m$ is the Euclidean ball in $\mathbb{R}^m$. Let us remark that the Morse index, even in the case $m>1$, can be related to the multiplicity of the solutions via the min-max theory (see e.g. \cite{solimini, lions}). Furthermore  the solutions and their geometrical features can also be studied. In fact, for a given solution $\phi$, define the nodal set $Z\subset\Omega$ as the closure of $\{x\in\Omega: \phi(x)=0\}$. It can be shown \cite{yang} that the Morse index provides an upper bound to the $L^p$ norm of the solution, to the Hausdorff measure of $Z$ and to the vanishing order of $x\in Z$; in our case this provides an a priori estimate of such properties of the solutions, since the upper bound on $i_L$ depends only on $C=\sup_{\phi\in\mathbb{R}}|V''(\phi)|$ and on the domain $\Omega$.
\section*{Appendix}
We present here a simple, self contained argument, for the proof of \emph{Proposition 2}.
\begin{prop} Let $Y$ be a real Hilbert space, $X$ a Banach space, $D$ an open subset of $X$. Assume that $F:D\times Y\to Y$ is a $C^1$ map, and that for every $x\in D$ the $x-$section of $F$, i.e. the map $F(x,\cdot):Y\to Y$ given by $y\mapsto F(x,y)$ is strongly monotone, that is, there exists $\mu_x>0$ such that
\[\langle F(x,u)-F(x,v), u-v\rangle\ge\mu_x\,\|u-v\|^2\quad\text{for every}\quad u,v\in Y.\]
Then the zero--set $Z(F)=\{(x,y)\in D\times Y:\,F(x,y)=0\}$ is the graph of a $C^1$ map $\varphi:D\to Y$.\end{prop}
\begin{proof} The point $x\in D$ being given, denote for simplicity by $f$ the $x-$section of $F$, that is, $f(u)=F(x,u)$ for every $u\in Y$, and let $\mu=\mu_x $. It is immediate to see that for all pairs $u,v\in  Y$ we have $\|f(u)-f(v)\|\ge\mu\,\|u-v\|$. This clearly shows that $f$ is injective with a Lipschitz continuous inverse $g:f(Y)\to Y$; moreover $f$ maps closed subsets of $Y$ into closed subsets of $Y$: in fact if $v_k=f(u_k)$ converges to $v\in Y$, then $u_k=g(v_k)$ is a Cauchy sequence in $Y$; if $u$ is its limit, then $v=f(u)$ by continuity of $f$. In particular, $f(Y)$ is closed in $Y$. Next observe that for every $c\in Y$ the differential $df(c)$ of $f$ at $c$ is also strongly monotone (with the same constant $\mu$). In fact, given $h\in Y$  we have, for every $t\in\mathbb R$:
\[\langle f(c+th)-f(c),th\rangle\ge\mu\,\|t\,h\|^2;\]
 assuming $t$ non zero:
 \[\left\langle\dfrac{f(c+th)-f(c)}t,h\right\rangle\ge \mu\,\|h\|^2,\]
 and when $t\to 0$ we get
 \[\langle df(c)[h],h\rangle\ge  \mu\,\|h\|^2\]
  as desired. 

By what we have just proved the range $df(c)[Y]$ is then a closed subspace of $Y$; and if $v$ is orthogonal to it we get
\[0=\langle df(c)[v],v\rangle\ge\mu\,\|v\|^2,\]
so that  $v=0$. Thus $df(c)$ is onto, and hence an isomorphism of $Y$ onto itself. The local inversion theorem is then applicable to $f$, proving that it is a local diffeomorphism at every point of $Y$; then $f(Y)$ is open, and being also closed it coincides with $Y$. We have proved that for every $x\in D$ the $x-$section $F(x,\cdot)$ is a  $C^1$ self--diffeomorphism of $Y$.  The conclusion is now immediate: for every $x\in D$ let $\varphi(x)$ be the unique zero of $y\mapsto F(x,y)$; the implicit function theorem (see, e.g. \cite{AP}) proves that $\varphi$ is of class $C^1$.\end{proof}

\par\noindent
{\bf Remark\,5.} If $X=\mathcal U_N$, $Y=\mathcal V_N$ and $L=u$, in Section 2 we have the following setting: a scalar valued function $u:X\times Y\to\mathbb R$, twice continuously differentiable in the Fr\'echet sense, such that the second differential with respect to $Y$ is strongly positive, i. e. there is $\mu>0$ such that
\[\partial^2_Yu(x,y)[h,h]\ge\mu\,\|h\|^2,\]
and moreover $\lim_{\|y\|\to\infty}u(x,y)=+\infty$ (also implied from this positivity). It is well--known that these hypotheses imply that the $x-$section of $u$ attains its absolute minimum in a unique point $\varphi(x)\in Y$, for every $x\in X$. The proof of this genuinely variational feature is based on the weak compactness of closed bounded convex subsets in a Hilbert space; this point $\varphi(x)$ is clearly the only singular point of the differential $\partial_Yu(x,y)\in Y^\ast$, that is, the unique point $\varphi(x)\in Y$ such that $\partial_Yu(x,\varphi(x))=0$. This implies, fully inside a variational scenario, the existence and uniqueness of $\varphi$. However, even though this argument is undoubtedly elegant,  smoothness of $\varphi$ does not follow, and this technique does not seem to be simpler than the one proposed in the previous Sections, which involves only elementary analysis, and no weak compactness arguments.

\bigskip

\par\noindent
{\bf Acknowledgements.}\ We thanks Giulia Antinori, Marco Degiovanni and Luigi Salce for fruitful conversations and suggestions.

\end{document}